\documentclass[journal]{IEEEtran}
\pdfoutput=1
\usepackage{amsmath,amsfonts,amssymb}
\usepackage{bm,times}

\usepackage{cite}

\usepackage{subfigure}
\usepackage{graphicx}
\usepackage{epstopdf}

\usepackage{amsthm}

\usepackage{algorithm}
\usepackage{algpseudocode}

\usepackage{etoolbox}
\apptocmd{\sloppy}{\hbadness 10000\relax}{}{} 

\usepackage[dvipsnames]{xcolor}

\usepackage{hyperref}
\hypersetup{
    colorlinks=true,
    linkcolor={blue!50!black},
    filecolor=magenta,
    urlcolor=NavyBlue,
    citecolor=NavyBlue,
}

\newtheorem{theorem}{Theorem}

\newtheorem{assume}[]{Assumption}

\begin{document}

\title{Set-Estimation based Networked Model Predictive Control for Energy Management of Faulty Microgrids}

\author{Quanwei~Qiu*,
       ~Fuwen~Yang,~\IEEEmembership{Senior Member,~IEEE,}
        and Yong~Zhu,~\IEEEmembership{Senior Member,~IEEE}%
\thanks{ * Corresponding author: Quanwei Qiu, E-mail: quanwei.qiu@outlook.com.}
\thanks{ Q. Qiu, F. Yang and Y. Zhu are with the School of Engineering and Built Environment, Griffith University, Gold Coast, 4222, Australia (e-mail: quanwei.qiu@outlook.com, fuwen.yang@griffith.edu.au, y.zhu@griffith.edu.au).}
}

\maketitle

\begin{abstract}
This paper addresses the issue of power flow control for partially faulty microgrids. In microgrid control systems, faults may occur in both electrical and communication layers. This may have severe effects on the operation of microgrids.
In addition, disturbances always coexist with faults in microgrids, which may further deteriorate system performance. To address the faults and disturbances simultaneously, a model predictive control (MPC) method based on set-membership estimation (SME) that transmits information via a communication network is proposed.
When electrical devices are nonfunctional or communication failures occur, the corresponding system states will become unavailable.
To this end, the SME method is employed to estimate the states with the existence of unknown-but-bounded process and measurement disturbances.
The networked MPC method is designed to schedule the power dispatch by using the forecasts of photovoltaic (PV) generation and load demand.
With these two methods, the fault-tolerant control can be achieved. Further, a deviation compensation method is proposed to compensate for the forecast errors. The effectiveness of the proposed control strategy is demonstrated through wireless communication tests using Raspberry $Pi$s.
\end{abstract}


\markboth{MANUSCRIPT}{}

\section{Introduction}
\IEEEPARstart{A}{n} increasing number of renewable energy sources (RESs), such as solar panels and wind turbines, are being integrated into power grids for environmental, economical, and technical reasons~\cite{olivares2014tsg}.
To manage the distributed energy resources, storage devices, and loads efficiently, microgrid technology was proposed in~\cite{lasseter2001pesw} and has since been widely studied and used.
However, this integration of RESs has introduced some challenges to the operation of microgrids.
Typically, the intermittent generation of RESs may cause power fluctuations in microgrids. The inertia of microgrid systems is reduced with the penetration of RESs, because RESs usually use power electronic converters as the interface with the power grid, which differs from traditional synchronous generators.
To address these two issues, battery storages are commonly used to diminish the power fluctuation and enhance the inertia~\cite{teleke2009tec,kim2016tps}.
Hence, it is necessary to develop proper energy management strategies to coordinate the power flow in microgrids. These kinds of control strategies correspond to the tertiary control of the microgrid hierarchical structure, which is responsible for the power flow control between the different generation and storage units within each microgrid and between microgrids and the utility grid~\cite{bidram2012tsg,olivares2014tsg}.

In addition to the issues of intermittency and low inertia caused by RESs, possible faults in microgrids may also severely affect system operation. Therefore, faults must be treated seriously in the control design of microgrids.
In the literature, researchers have studied many faults in the design of the three control levels of microgrid hierarchical structure.
Generally, faults in microgrid systems can be categorized as four types according to their location in microgrid systems: actuator faults, sensor faults, communication faults, and plant faults~\cite{Shahab2020tsg}.
Reference~\cite{Shahab2020tsg} proposes a distributed secondary control for islanded microgrids with actuator faults.
Reference~\cite{morato2020ijepes} presents an energy management strategy for the tertiary control level to deal with the actuator faults.
Sensor fault-tolerant control strategies are designed for the voltage source converter controller in the primary control level of microgrids~\cite{gholami2016pesgm,gholami2018ijepes}, whereas sensor and actuator faults are addressed simultaneously to restore the voltage and frequency of secondary control in~\cite{afshari2020tps}.
Communication failures are taken into account in the design of a microgrid energy management system in~\cite{loser2019fallback}, and an unreliable communication network is considered in the cooperative secondary control with the existence of actuator faults in~\cite{afshari2020tsg}.
Compared with the three types of faults described above, plant faults have attracted the most attention in the literature, as they are related to common electrical faults in power systems~\cite{genduso2010icem,jin2018tie,hosseinzadeh2018tsg,prodan2014ijepes,prodan2015energy}.
Most of these works focus on the faults occurring in power devices---for example, electronic converters~\cite{genduso2010icem,jin2018tie}, energy storage systems~\cite{hosseinzadeh2018tsg}, and renewable generators~\cite{prodan2014ijepes,prodan2015energy,hosseinzadeh2018tsg}.

However, to the best of the authors’ knowledge, none of the literature has investigated plant faults and communication faults together for the tertiary control of microgrids while also taking into account disturbances.
In practice, a microgrid may become partially faulty, which may be caused by disconnected transmission lines or nonfunctional power devices.
Since RESs are usually small-scale and distributed geographically, communication between microgrid components is necessary for power-sharing management~\cite{ouammi2015tsg,yan2019tsg}.
Yet communication networks may be unreliable, as they are vulnerable to faults (or even failures).
Moreover, disturbances in system process or measurement output can significantly deteriorate system performance if proper treatment is not administered.
Therefore, it is essential to develop a method for addressing faults and disturbances simultaneously.

When battery storage within microgrids is nonfunctional or its corresponding communication link is down, its state will become unavailable.
To estimate the state, numerous state estimation approaches have been used in the existing literature \cite{zhang2014tie,dai2015am,shen2010am,ge2019tc1148}. These studies have mainly adopted two approaches: Kalman filtering and $H_\infty$ filtering. Yet, both these estimation approaches are still point-wise and focus on optimizing the estimation error. The statistical properties of disturbances are usually required in these estimation approaches. To overcome these limitations, an alternative set-membership filtering, which calculates a bounding ellipsoidal set that encloses all possible state estimations with the assumption that the noises are unknown-but-bounded, has been proposed and extensively investigated in many works \cite{witsenhausen1968tac,yang2009am,yang2009tac,qiu2020cta}.

One of the most effective methods for dealing with the faults that occur in RES-based microgrids is model predictive control (MPC). This is because the MPC technique has the advantage of employing forecasts conveniently, which is significant in terms of addressing the intermittency of RES generation.
By using the forecasts of RES generation and load demand, MPC predicts a sequence of control inputs that can optimally schedule the power dispatch for the faulty microgrids~\cite{brenna2018tsg}. In addition, the control sequence can be used as default set-point power values for RESs and batteries. This enables the electrical devices to overcome communication failures by implementing the precomputed default values when there is no new control signal update.
Further, MPC can easily manipulate system constraints and is especially efficient in handling multi-variable constrained control problems and overcoming the uncertainty, nonlinearity, and correlation of system process. As a result, numerous studies have been conducted on using MPC to address the power dispatch problem of microgrids~\cite{worthmann2015tsg,braun2016tac,Halvgaard2016tsg,shan2020tsg}.

Inspired by these factors, this paper proposes a microgrid power dispatch (MPD) strategy that employs the set-membership estimation (SME)-based MPC.
The MPC is communication-based and designed in a centralized way: a communication network transmits information between the central controller and local units.
Hence, this paper refers to the MPC method as networked MPC and considers the communication faults in this communication network.
The SME is used to estimate the unavailable battery states caused by faults while considering unknown-but-bounded process disturbances and measurement noises.
It is worth noting that the electrical fault (e.g., electricity transmission line disconnection) does not imply the communication failure.
That means although the electrical connections between microgrid components fail in some scenarios, the communication network may still operate well and the component states may still be available.
Therefore, the connection status of the electrical grid and the communication network should be considered separately in the system modeling.
In addition, in most of the previous literature~\cite{parisio2014tcst,hans2019tse}, the forecasts for the RES generation and load demand are usually assumed to be certain. This assumption of correct prediction is reasonable to some extent, since one-day-ahead weather predictions to calculate RES generation power can be fairly accurate, and the residential loads tend to follow daily patterns~\cite{worthmann2015tsg}.
However, possible deviations from the forecasts cannot be ignored in practical application.
To handle this issue, a one-step-ahead compensation method is proposed to deal with the forecast errors.

In summary, the main contributions of this paper can be described as follows:
1) It proposes a novel system modeling method that incorporates the connection status of both the electrical grid and the communication network.
2) It designs the SME method to estimate the unavailable system states caused by faults and to handle the unknown-but-bounded disturbances.
3) It develops a power dispatch strategy using networked MPC that can maximize the utilization of RESs while satisfying the load demand under faulty conditions.

The remainder of this paper is organized as follows. Section \MakeUppercase{\romannumeral 2} models the investigated system. Based on this model, Section \MakeUppercase{\romannumeral 3} derives the ellipsoidal state estimation, and Section \MakeUppercase{\romannumeral 4} gives the MPC problem formulation. Section \MakeUppercase{\romannumeral 5} presents the deviation compensation approach and corresponding control algorithm. Section \MakeUppercase{\romannumeral 6} shows the simulation test results, and Section \MakeUppercase{\romannumeral 7} concludes the paper.

$\mathbf{Notation}$: Throughout this paper, $\mathbb{R}$ and $\mathbb{R}^n$ denote, respectively, the set of real numbers and the $n$ dimensional Euclidean space. The superscript ``$T$'' represents the transpose for a matrix. $|\cdot|$ stands for any absolute value of real numbers. $\textnormal{diag}(v)$ is used to create a square diagonal matrix with the elements of vector $v$ on the main diagonal.
The symbol $\wedge$ is defined as the AND logic operation of two inputs.
All the mathematical operations---comparison, minimization, maximization, and logic operation---are performed element-wise for the vectors.

\section{System Modelling}\label{sec:1}

\subsection{System description}
This work considers a multi-microgrid system (MMS) that is completely supplied by photovoltaic (PV) generators, as shown in Fig.~\ref{fig:2.1}.
Each solar panel is equipped with a battery to store the excess generated power and provide power to the load when required.
The central controller on the right-hand side of Fig.~\ref{fig:2.1} is the power dispatch controller that coordinates the charging and discharging schedule of the battery storages.
The solid lines in the figure represent the power transmission lines, while the dashed lines represent the communication network.
In each microgrid, a hub gathers information from the local microgrid units and then informs the central controller. It also distributes the control signals received from the central controller to local units via the communication network.

In the rest of this paper, similar to the adjacency matrix in graph theory, connection vectors are used to describe the connection status of the electrical grid and the communication network. All the elements in the connection vectors are Boolean, which are either $0$ or $1$.
$1$ means there is no fault, while $0$ means the electrical device is nonfunctional (or just disconnected from the grid) or the communication fails.

Note that the connection links within the communication network are bidirectional, which means if the value is $0$, the central controller cannot provide set-points to the unit, nor can it receive new measurements from the unit.
If a microgrid component is nonfunctional, both of its corresponding elements in the electrical and communication connection vectors will be set as $0$, since the component does not have power flow and cannot provide its state information in this faulty situation.
\begin{figure}[htbp]
  \centering
  \includegraphics[width=0.48\textwidth]{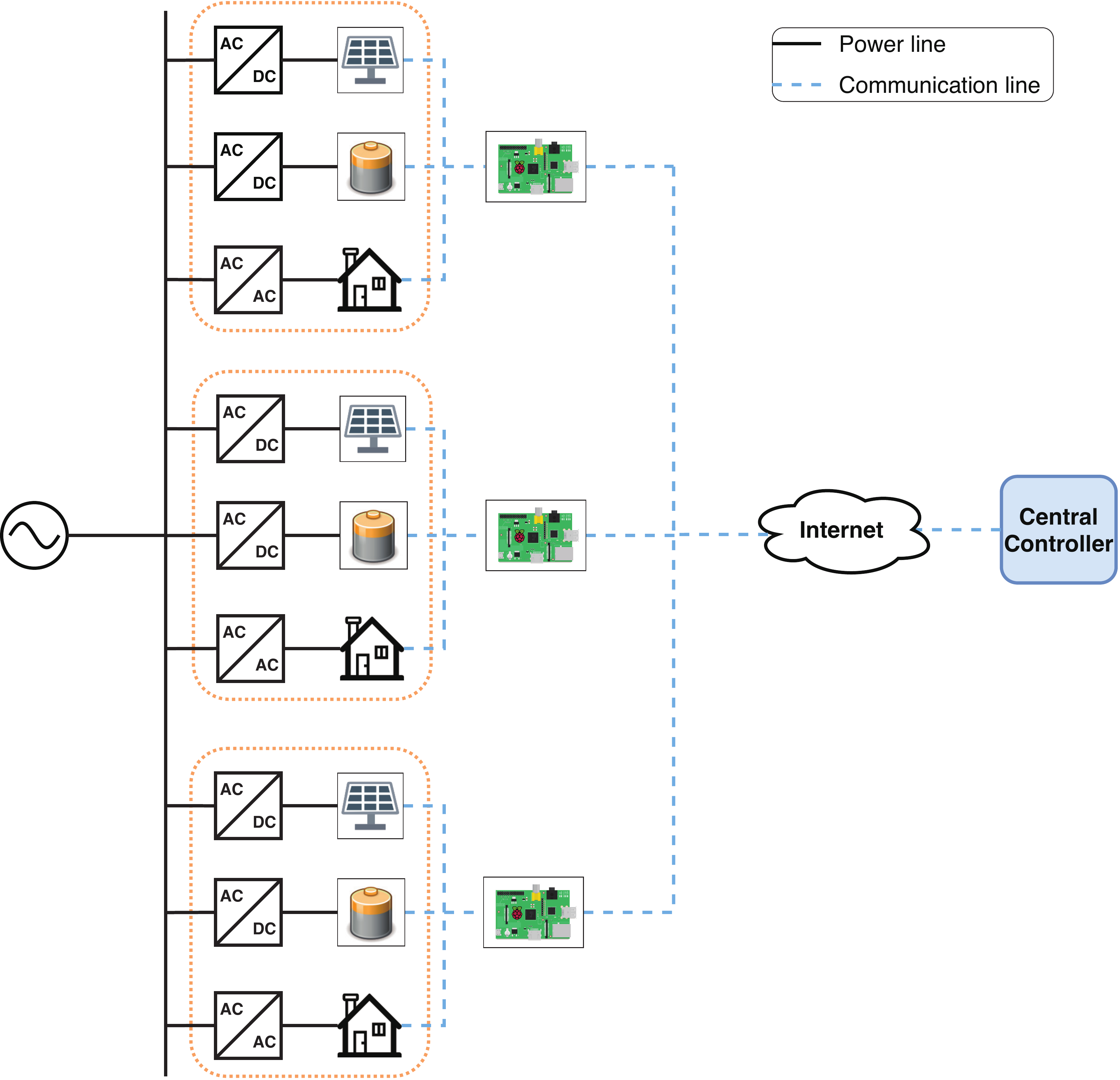}
  \caption{Multiple microgrids with a power dispatch controller}\label{fig:2.1}
\end{figure}

\subsection{Dynamic model of battery storages}
As in the MMS in Fig. \ref{fig:2.1}, the system is supplied by $n_s\in \mathbb{Z}^+$ PV generators, and there are $n_b\in \mathbb{Z}^+$ battery storage units. In total, there are
$n_g=n_b+n_s \in \mathbb{Z}^+$ distributed units connected in the MMS, exclusive of the load units.
The discrete-time model for the battery of the $i$th subsystem is given by the following:
\begin{equation}\label{eq:1.1}
  x_i(k+1)=x_i(k)-T_sP_{b,i}(k),
\end{equation}
where $i\in \{1,\cdots, n_b\}, n_b\leq n_g$, $x_i(k)$ is the absolute state of charge (SoC) of the $i$th battery at time step $k, k\in \mathbb{N}$, $P_{b,i}(k)$ is the battery charging/discharging power, and $T_s$ denotes the sampling time in hours ($h$). Note that the target system is expressed in a per-unit system, which means some elements in the system are represented by per-unit values. Thus, the power quantities are labeled with the symbol $pu$, while the absolute SoC is labeled with the symbol $puh$.

Thus, the global system model can be described as follows:
\begin{equation}\label{eq:1.2}
  \left[
  \begin{matrix}
    x_1(k+1) \\
    x_2(k+1) \\
    \vdots   \\
    x_{n_g}(k+1)
  \end{matrix}\right]=\left[\begin{matrix}
                              x_1(k) \\
                              x_2(k) \\
                              \vdots \\
                              x_{n_g}(k)
                            \end{matrix}\right]
                            -T_s\left[\begin{matrix}
                                                     P_{b,1}(k) \\
                                                     P_{b,2}(k) \\
                                                     \vdots \\
                                                     P_{b,n_g}(k)
                                                   \end{matrix}\right],
\end{equation}

Further, model (\ref{eq:1.2}) can be rewritten into a compact form as follows:
\begin{equation}\label{eq:1.3}
  x(k+1)=x(k)-T_sP_b(k),
\end{equation}
where $x(k)\in\mathbb{R}^{n_b}$, $P_b(k)\in\mathbb{R}^{n_b}$.

\subsection{Power outputs of battery storages and PV generators}

For the battery storages, the charging/discharging power is defined as follows:
\begin{subequations}\label{eq:1.5}
\begin{gather}
\hat u_b(k)=(I-\textnormal{diag}(\mathcal{A}_b(k)))u_b^d(k)+\textnormal{diag}(\mathcal{A}_b(k))u_b(k),\\
 P_b(k)=\textnormal{diag}(\mathcal{G}_b(k))\hat u_b(k) \label{eq:1.5b},
\end{gather}
\end{subequations}
where all the power values are written in a compact form,
$u_b^d(k)\in \mathbb{R}^{n_b}$ is the default charging/discharging power of the batteries, $u_b(k)\in \mathbb{R}^{n_b}$ is the power value to be determined, and $\hat u_b(k)\in \mathbb{R}^{n_b}$ is the set-point power value provided to the batteries. $\mathcal{A}_b(k)\in \mathbb{R}^{n_b}$ is a vector that represents the communication status between the batteries and the central controller with all the elements being Boolean, where $\mathcal{A}_{b,i}(k)=1, i\in\{1,\cdots, n_b\}$ indicates normal communication and $\mathcal{A}_{b,i}(k)=0$ implies a communication failure in the $i$th battery. $\mathcal{G}_b(k)\in \mathbb{R}^{n_b}$ is a vector that implies the electrical connection status of the batteries in the MMS. All the elements of $\mathcal{G}_b(k)$ are also Boolean. If $\mathcal{G}_{b,i}(k)=1$, then the $i$th battery has no fault; otherwise, the $i$th battery is nonfunctional or disconnected from the grid.

For the PV generators, the actual power output is obtained as follows:
\begin{subequations}\label{eq:1.4}
\begin{gather}
 \hat u_s(k)=(I-\textnormal{diag}(\mathcal{A}_s(k)))u_s^d(k)+\textnormal{diag}(\mathcal{A}_s(k))u_s(k),\\
   P_s(k)=\textnormal{diag}(\mathcal{G}_s(k))\textnormal{min}(\hat u_s(k), P_s^a(k))\label{eq:1.4b},
 \end{gather}
\end{subequations}
where $u_s^d(k)\in \mathbb{R}^{n_s}$ is the default set-point power value for the PV generators, $u_s(k)\in \mathbb{R}^{n_s}$ is the generation power to be determined, $\hat u_s(k)\in \mathbb{R}^{n_s}$ is the set-point power value sent to the PV generators, $P_s^a(k)\in \mathbb{R}^{n_s}$ is the maximum available output power that the PV generators can provide at time step $k$, and $P_s(k)\in \mathbb{R}^{n_s}$ is the actual power value generated by the PV generators.
$\mathcal{A}_s(k)\in \mathbb{R}^{n_s}$ and $\mathcal{G}_s(k)\in \mathbb{R}^{n_s}$ are vectors that describe the communication and electrical connection status between the PV generators and the central controller, respectively.

In addition, the PV generation power values forecasted at time step $k$ are denoted as $\hat P_s^a(k+n|k), n\in \mathbb{N}$ and the load demand values as $\hat P_l(k+n|k)$.
Usually, it is assumed that there is no deviation between the actual values and the forecasted values for the time step $k$---that is:
\begin{subequations}\label{eq:1.6}
\begin{gather}
  P_s^a(k)=\hat P_s^a(k),\\
  P_l(k)=\hat P_l(k).
\end{gather}
\end{subequations}

\begin{figure}[htbp]
  \centering
  \includegraphics[width=0.3\textwidth]{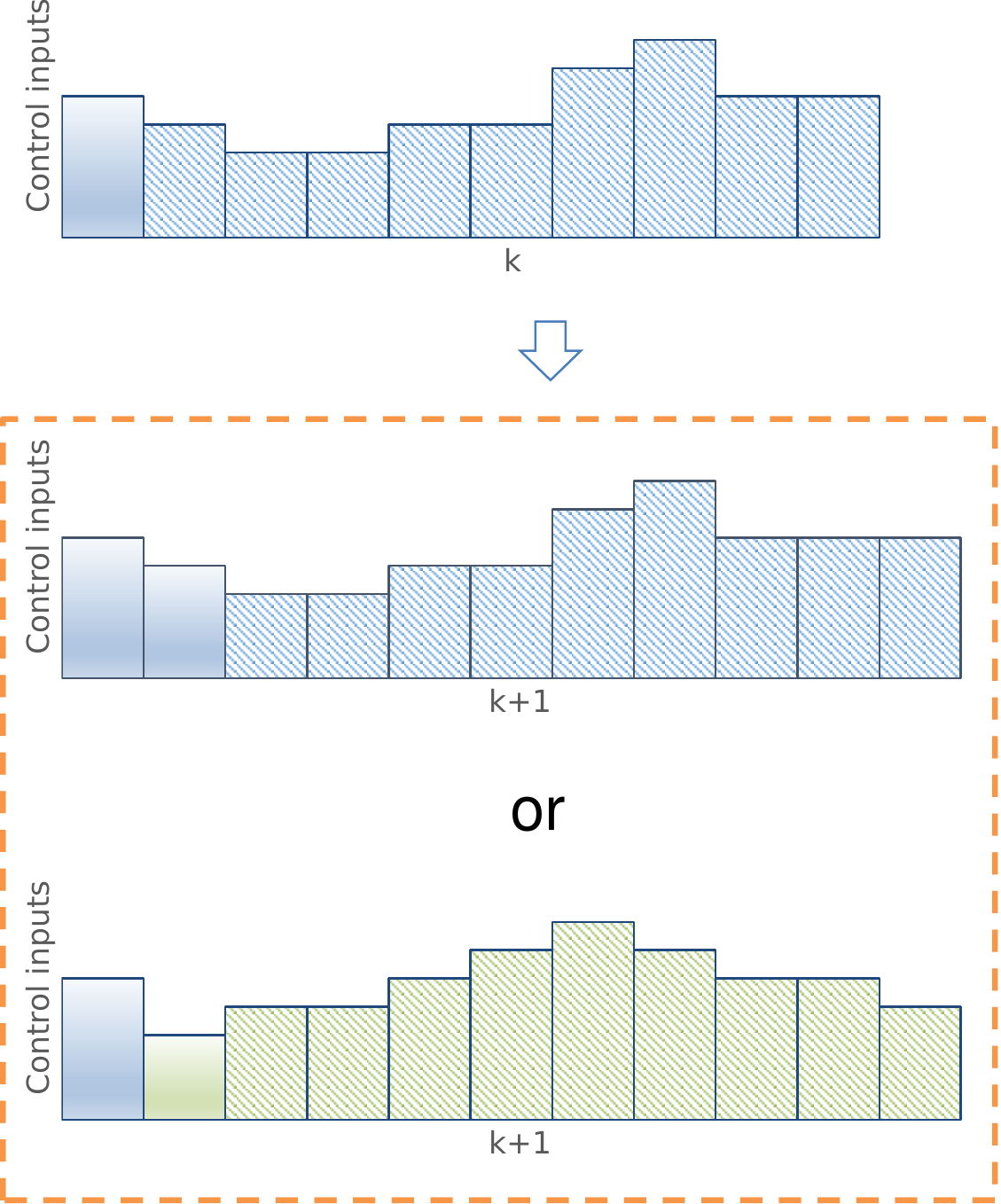}
  \caption{Control input sequence update of local units}\label{fig:2.2}
\end{figure}

\subsection{Communication-based control sequence update}
As illustrated in first picture of Fig.~\ref{fig:2.2}, assuming that there is communication between a local unit and the central controller at time $k$, the local unit can receive a control sequence from the controller.
At the next time step, $k+1$, if the communication is disconnected, the local unit will implement the obtained control sequence from last step, $k$, as the default control, which is shown in the second picture; if there is no communication problem, then the local unit will receive a new control sequence from the central controller and update its future default control inputs, as depicted in the last picture.

With this updated setup, the future default control sequence determined at time step $k$ can be formulated as follows:
\begin{align}\label{eq:1.7}
   u_j^d(k+n|k)=\left\{
   \begin{aligned}
   &u_j(k+n|k-t), n\in\mathbb{Z}_{[0,N-t]}\\
   &u_j(k-t+N|k-t), n\in\mathbb{Z}_{[N-t+1,N]}
   \end{aligned}
   \right.
 \end{align}
where $j\in\{1,\cdots,n_g\}$, $k-t$ is the last time step before the failure occurs, and $u_j(\cdot|k-t)$ is the control sequence computed by solving the MPC optimization problem at time $k-t$.

\section{Ellipsoidal Set-Membership State Estimation}\label{sec:3}

In this section, an SME method is designed to estimate the battery states while also taking into account process disturbances and measurement noises.

The system model (\ref{eq:1.3}) in Section \ref{sec:1} is an ideal model based on the assumption that the system state can be obtained accurately all the time.
However, in a practical situation, this assumption is always arbitrary, as the system state can be inaccurate or even unavailable. This is especially likely when the system is influenced by process disturbances and measurement noises and communication failure occurs.
To describe the system more precisely, the following disturbed system model is considered:
\begin{align}\label{eq:3.1}
  \left\{
  \begin{aligned}
  x(k+1)&=x(k)-T_s(P_b(k)+\omega(k))\\
  y(k)&=\textnormal{diag}(\mathcal{A}_b(k))(x(k)+\upsilon(k))
  \end{aligned}
  \right.,
\end{align}
where $\omega(k)\in\mathbb{R}^{n_b}$ represents the process disturbance that affects the battery state, which can be caused by the charge/discharge dissipation or power loss of the battery, and $\upsilon(k)\in\mathbb{R}^{n_b}$ represents the measurement noise that affects the accuracy of the measurement output.

Further, substituting (\ref{eq:1.5}) into (\ref{eq:3.1}), the following system model is obtained:
\begin{align}\label{eq:3.5}
  \left\{
  \begin{aligned}
  x(k+1)&=x(k)+B(k)u(k)+F(k)\omega(k)+\delta(k)\\
  y(k)&=C(k)x(k)+D(k)\upsilon(k)
  \end{aligned}
  \right.,
\end{align}
where $u(k)=u_b(k)$, $B(k)=-T_s\textnormal{diag}(\mathcal{G}_b(k))\textnormal{diag}(\mathcal{A}_b(k))$, $F(k)=-T_s$, $C(k)=D(k)=\textnormal{diag}(\mathcal{A}_b(k))$, and $\delta(k)=-T_s\textnormal{diag}(\mathcal{G}_b(k))*(I-\textnormal{diag}(\mathcal{A}_b(k)))u_b^d(k)$ is the already known constant related to the default battery output power at every time step.

As in practical systems, disturbances are always bounded in some specific ranges. Therefore, the following Assumption~\ref{as:3.1} is given to define the boundary of the disturbances. Assumption~\ref{as:3.2} defines the boundary of the initial system state, which is essential for the recursive state estimation method.

\begin{assume}\label{as:3.1}
  $\omega(k)$ and $\upsilon(k)$ are unknown but bounded by the following two ellipsoids:
\begin{align}
  \omega(k)\in \varepsilon(\mathbf{0},Q(k))&=\{\omega(k)|{\omega^T(k)} {Q^{-1}(k)}\omega(k)\leq 1\}\nonumber\\&\triangleq W(k),\label{eq:3.2}\\
  \upsilon(k)\in \varepsilon(\mathbf{0},R(k))&=\{\upsilon(k)|{\upsilon^T(k)} {R^{-1}(k)}\upsilon(k)\leq 1\}\nonumber\\&\triangleq V(k).\label{eq:3.3}
\end{align}
where $Q(k)={Q^T(k)}>0$ and $R(k)={R^T(k)}>0$ are known matrices with compatible dimensions. $Q(k)$ and $R(k)$ determine how far the ellipsoids extend in every direction from the origin.
\end{assume}

\begin{assume}\label{as:3.2}
  The initial state is bounded by a given ellipsoid
  \vspace{-0.2cm}
\begin{align}\label{eq:3.4}
  x(0)\in & \ \varepsilon(\hat{x}(0),P(0))\nonumber\\
=&\{x(0)|(x(0)-\hat{x}(0))^TP^{-1}(0)(x(0)-\hat{x}(0))\leq 1\}\nonumber\\
\triangleq& X(0),
\end{align}
where $\hat{x}(0)$ is an estimation of $x(0)$ assumed to be given, and $P(0)=P^T(0)>0$ is a known matrix.
\end{assume}

The system state is estimated through a prediction step and a measurement update step.

Prediction step: use the known state estimate and control input of last step to predict the current state estimate as follows:
\begin{align}\label{eq:3.6}
\left\{
\begin{aligned}
  \hat{x}({k|k-1})=&\hat{x}({k-1|k-1})+B({k-1})u({k-1})\\
  &+\delta(k-1)\\
  \hat{y}({k|k-1})=&C(k)\hat{x}({k|k-1}).
\end{aligned}
\right.
\end{align}
Measurement update step: use the current system output to update the current state estimate as follows:\vspace{-0.1cm}
\begin{align}\label{eq:3.7}
  \hat{x}({k|k})&=\hat{x}({k|k-1})+L({k})(y({k})-\hat{y}({k|k-1})).
\end{align}

By substituting (\ref{eq:3.6}) into (\ref{eq:3.7}), the current state estimate can be obtained as follows:
\vspace{-0.2cm}
\begin{align}\label{eq:3.8}
  \hat{x}({k})=&(I-L(k)C(k))\hat{x}({k-1})+(I-L(k)C(k))B({k-1})\nonumber\\
  &u({k-1})+L(k)y(k)+(I-L(k)C(k))\delta(k-1).
\end{align}
Note that $\hat{x}({k|k})=\hat{x}({k})$ and $\hat{x}({k-1|k-1})=\hat{x}({k-1})$.

Based on the two-step estimation strategy noted above, the following theorem can be obtained to provide sufficient conditions for the existence of the state estimation ellipsoid in which the estimated state resides.
\begin{theorem}\label{th:3.1}
For system (\ref{eq:3.5}), suppose that the state of the last step, $x(k-1)$, lies in its state estimation ellipsoid $\varepsilon(\hat{x}(k-1),P(k-1))$. Then, the state of the current step $x(k)$ resides in its state estimation ellipsoid $\varepsilon(\hat{x}(k),P(k))$, if there exist matrices $P(k)>0$, $L(k)$, and nonnegative scalars $\lambda_1$, $\lambda_2$ and $\lambda_3$ such that
  \begin{align}\label{eq:3.9}
  \left[
  \begin{matrix}
    -\Psi_{\lambda_1,\lambda_2,\lambda_3} & \Phi_{\eta}^T(k) \\
    \Phi_{\eta}(k) & -P(k)
  \end{matrix}
  \right]\leq 0,
  \end{align}
  where
  \vspace{-0.2cm}
  \begin{multline*}
  \Phi_{\eta}(k)=\left[
  \begin{matrix}
    (I-L(k)C(k))E({k-1})
    \end{matrix}\right.\\
    \left.\begin{matrix}
    & (I-L(k)C(k))F(k-1) & -L(k)D(k) & 0
  \end{matrix}
  \right],
  \end{multline*}
  $\Psi_{\lambda_1,\lambda_2,\lambda_3}=\textnormal{diag}(\lambda_1I, \lambda_2Q^{-1}({k-1}),\lambda_3R^{-1}(k), 1-\lambda_1-\lambda_2-\lambda_3)$.
\end{theorem}
\begin{proof}
  See the Appendix.
\end{proof}

The inequality (\ref{eq:3.9}) provides a sufficient condition to obtain the state estimation ellipsoids recursively.
To minimize the volumes of the ellipsoids, a convex optimization approach is derived to find the minimal ellipsoids---that is, $P(k)$ and $L(k)$ are obtained by solving the following optimization problem:\vspace{-0.2cm}
\begin{gather}\label{op:3.1}
 \min_{P(k),L(k),\lambda_1,\lambda_2,\lambda_3} \textnormal{trace }P(k),\\
 \textnormal{subject to }\quad(\ref{eq:3.9}). \nonumber
\end{gather}

Thereafter, the state estimate of the current step, $\hat{x}(k)$, can be calculated using (\ref{eq:3.8}) with the state estimate of the last step $\hat{x}({k-1})$, control input $u({k-1})$, the output measurement $y(k)$, and the measurement update gain $L(k)$ of the current step computed by (\ref{op:3.1}).

\section{MPC Design}

As the system model is now constructed and the states are estimated, the next step is to calculate the optimal control sequence as inputs that feed back to the system via the MPC technique. Hence, in this section, constraints are given to restrain the signals of different system components, cost functions are designed to emphasize the control objectives by penalizing different terms, and the optimization problem for MPC is derived.

\subsection{Component constraints}
For the MPC formulation, the following constraints are imposed on different components:
\begin{subequations}\label{eq:2.1}
\begin{equation}
   \begin{split}
  \textnormal{Solar components: } P_s^{min}\leq P_s(k+n) \leq P_s^{max},\\
  \textnormal{Battery components: } P_b^{min}\leq P_b(k+n) \leq P_b^{max},\\
  x^{min}\leq x(k+n) \leq x^{max},\\
   k,n\geq 0.\nonumber
   \end{split}
\end{equation}
\end{subequations}
Note that the predicted PV generation power, battery power, and SoC are denoted by $P_s(k+n|k)$, $P_b(k+n|k)$, and $x(k+n|k)$, respectively. For simplicity, these denotations are written as $P_s(k+n)$, $P_b(k+n)$, and $x(k+n)$, respectively.

In addition, a power balance equation must be satisfied.\vspace{-0.2cm}
\begin{multline}\label{eq:2.2}
  {\mathcal{G}_s(k+n)}^TP_s(k+n)+{\mathcal{G}_b(k+n)}^TP_b(k+n)\\
  +P_g(k+n)-{\mathcal{G}_l(k+n)}^TP_l(k+n)=0,
\end{multline}
where $\mathcal{G}_s(k+n)$, and $\mathcal{G}_l(k+n)$ are the electrical connection vectors for PV and load units, respectively. $P_g(k+n)$ is the power supplied to/by the utility grid.

If the microgrid is islanded---that is, the utility grid power $P_g(k+n)=0$---define $P_r(k+n)\in \mathbb{R}$ by $P_r(k+n)={\mathcal{G}_s(k+n)}^TP_s(k+n)-{\mathcal{G}_l(k+n)}^TP_l(k+n)$, then Equation (\ref{eq:2.2}) can be rewritten in the following form:
\begin{equation}\label{eq:2.3}
  \left[
  \begin{matrix}
    {\mathcal{G}_b(k+n)}^T  \\
    -{\mathcal{G}_b(k+n)}^T
  \end{matrix}\right]
  P_b(k+n)\leq \left[
  \begin{matrix}
    -P_r(k+n) \\
    P_r(k+n)
  \end{matrix}\right].
\end{equation}

Note that if any electrical fault occurs in the batteries---that is, $\mathcal{G}_{b,i}(k+n)=0$---the dimension of the above linear matrix inequality can be reduced by removing the $i$th element in the variable $P_b(k+n)$ and by removing $\mathcal{G}_{b,i}(k+n)$ from $\mathcal{G}_b(k+n)$, since the power value of the corresponding faulty battery is $0$ when the electrical fault occurs.

\subsection{Cost functions}
Different cost functions for the grid components are adopted to penalize different control objectives.
For PV generators, the following cost function is used to make the output power as large as possible to reduce energy waste:
\begin{subequations}\label{cos:1}
\begin{equation}\label{cos:1a}
  J_s(k+n)=-C_s^TP_s(k+n),
\end{equation}

For batteries, the following cost function is used to reduce high power operation and make the energy levels stay at a threshold:\vspace{-0.2cm}
\begin{multline}\label{cos:1b}
  J_b(k+n)=P_b^T(k+n)\textnormal{diag}(C_{b,1})P_b(k+n)\\+\Delta x^T(k+n)\textnormal{diag}(C_{b,2})\Delta x(k+n),
\end{multline}
where $\Delta x(k+n)=\textnormal{max}(x_b^{min}-x(k+n),0)+\textnormal{max}(x(k+n)-x_b^{max},0)$.
Note that the second term on the right-hand side of Equation (\ref{cos:1b}) is a soft constraint.

If the utility grid power is considered and not predefined, the following cost function is given to make the local microgrids rely less on the utility grid---that is, to reduce the power drawn from the utility grid or even feed power to the utility grid when there is excess power:
\begin{equation}\label{cos:1c}
  J_g(k+n)=P_g^T(k+n)\textnormal{diag}(C_{g,1})P_g(k+n)+C_{g,2}P_g(k+n),
\end{equation}
\end{subequations}
where $P_g$ is the utility grid power to be determined.

By summing up all these costs, the overall cost function is obtained as follows:
\begin{equation}\label{cos:2}
  J(k+n)=J_s(k+n)+J_b(k+n)+J_g(k+n).
\end{equation}

\subsection{Optimization problem}
Finally, with the constraints and cost functions defined as above, the MPC optimization problem can be achieved as follows:
\begin{gather}\label{op:4.1}
  \min_{u}\quad J_N(k),\\
  \textnormal{subject to }\quad (\ref{eq:1.3}),(\ref{eq:1.4}),(\ref{eq:1.5}),(\ref{eq:2.1}),(\ref{eq:2.3}) \nonumber
\end{gather}
where $J_N(k)=\sum_{n=0}^{N} J(k+n)$ is the sum of the overall cost function $J(k+n|k)$ over the future control horizon $N\in \mathbb{N}$, with $J(k+n|k)=J(k+n)$ given as in (\ref{cos:2}).

By using some optimization solvers, this optimization problem can be solved, and the optimal control sequence can be implemented in the target system.
In this way, the designed networked MPC strategy can address both the electrical faults and the communication failures simultaneously with the constructed system model.

\section{Control Algorithm with Deviation Compensation}
In a practical situation, the actual values of PV power generation and load demand are highly likely to deviate from the forecasted values. This may cause a power mismatch between the generators and loads, which will further affect the global power balance.
A reasonable approach to compensate for this deviation is to reschedule the battery charging/discharging power values.
To this end, a deviation compensation method is proposed to calculate a compensation term, which will be further added to the set-point power values calculated by MPC.
This procedure can be formulated as follows:
\begin{gather}\label{eq:5.1}
  u_b^*(k)=\hat{u}_b(k)+\textnormal{diag}({\mathcal{E}_b(k)})\lambda_b\sigma(k),
\end{gather}
where the second term on the right-hand side of this equation represents the additional values added to the set-point power values of batteries, and $\sigma(k)\in \mathbb{R}$ is a scalar to be determined according to the deviation.
$\mathcal{E}_b(k)$ is the connection vector for the batteries in which the elements indicate the connection status of both the electrical network and the communication network.
This connection vector has the following relationship with the previously defined electrical connection vector and the communication connection vector:
${\mathcal{E}_b(k)}=\mathcal{A}_b(k)\wedge \mathcal{G}_b(k)$.
If there is any electrical fault or communication failure, the corresponding element in the above vector will be set as $0$.
$\lambda_b$ is the coefficient vector for the batteries.

Obviously, the compensated set-point power values of the batteries $u_b^*(k)$ still need to satisfy the following system constraints:
\begin{gather}
P_b^{min}\leq u_b^*(k)\leq P_b^{max}, \label{eq:5.3}\\
   x^{min}\leq x(k)-T_su_b^*(k)\leq x^{max}. \label{eq:5.4}
\end{gather}

To calculate $\sigma(k)$, the following maximization problem is proposed to ensure that the system constraints, as noted above, can always be satisfied by incorporating the additional compensation term:
\begin{gather}
  \max \quad{|\sigma(k)|},\label{op:5.1}\\
\textnormal{subject to}\hspace{40mm} \nonumber\\
(\ref{eq:5.3}),(\ref{eq:5.4}), \textnormal{and} \nonumber\\
  \left\{
  \begin{aligned}
  0\leq\sigma(k)\leq \hat{\sigma}(k) \quad\quad \hat{\sigma}(k)\geq 0\\
  \hat{\sigma}(k)<\sigma(k)\leq 0 \quad\quad \hat{\sigma}(k)\leq 0
  \end{aligned}
  \right.,\nonumber
\end{gather}
where
\begin{gather}
  \hat \sigma(k)=\frac{{\mathcal{E}_s(k)}^T\Delta P_s(k)+{\mathcal{E}_l(k)}^T\Delta P_l}{{\mathcal{E}_b(k)}^T\lambda_b},\label{eq:5.2}\\
  \Delta P_s(k)=\min(\hat{u}_s(k),\hat{P}_s^a(k))-\min(\hat{u}_s(k),P_s^a(k)),\nonumber\\
  \Delta P_l(k)=P_l(k)-\hat{P}_l(k),\nonumber
\end{gather}
$\mathcal{E}_s(k)$ and $\mathcal{E}_l(k)$ are the connection vectors for the PV generators and loads, which have the same definition as $\mathcal{E}_b(k)$.

\begin{figure}[htbp]
  \centering
  \includegraphics[width=0.48\textwidth]{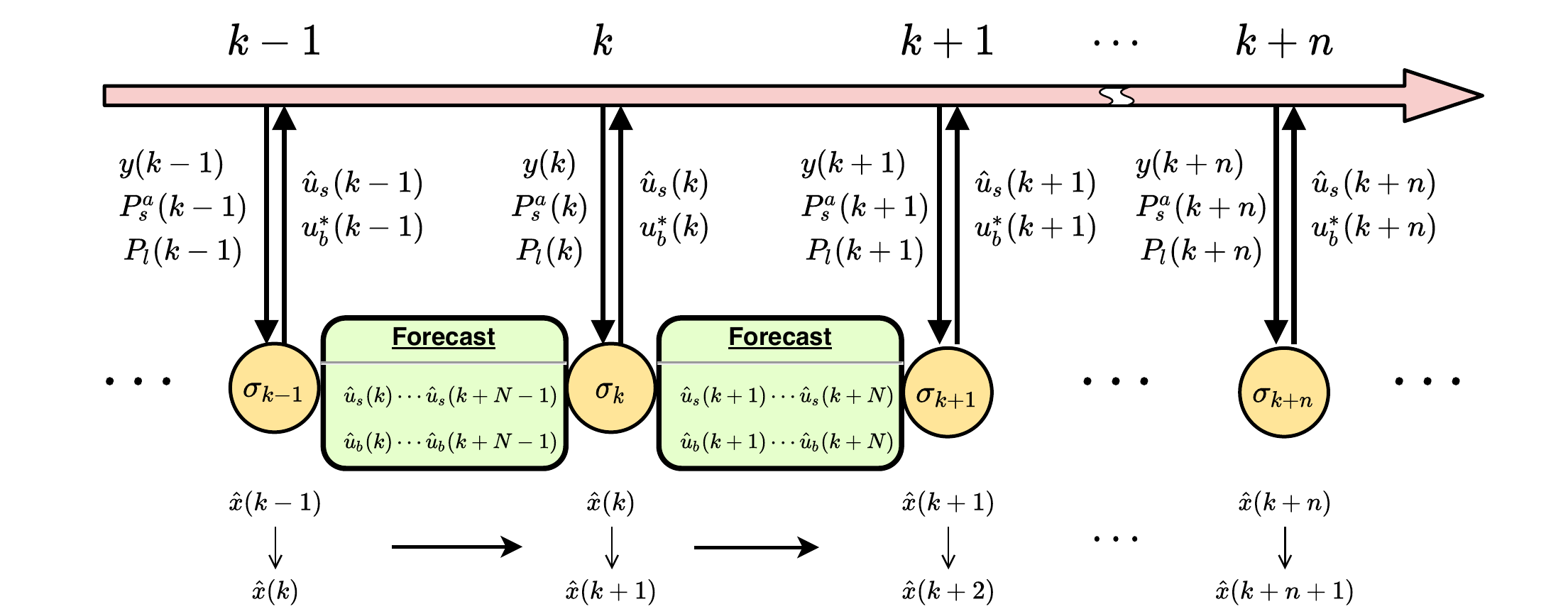}
  \caption{Control mechanism with deviation compensation}\label{fig:4.1}
\end{figure}

\begin{algorithm}
\caption{Recursive Power Dispatch Control Algorithm with Deviation Compensation}
\small
\label{alg:1}

\begin{algorithmic}
\State  \textbf{1. Initialization}\\
\quad   Given an initial input $\hat u_s(0)$ and $u_b^*(0)$, the initial state $\hat x(0)$, the prediction horizon $N$, and set $k=1$ and the predicted inputs for time step $k$ as $\hat{u}_s(k)=\hat{u}_s(0)$ and $\hat{u}_b(k)=u_b^*(0)$.

\\\hrulefill
\State  \textbf{2. Deviation Compensation}

\\1) Obtain the current maximum available power of the PV panels $P_s^a(k)$ and the load demand $P_l(k)$;

\\2) Calculate $\sigma(k)$ by (\ref{op:5.1});

\\3) Calculate $u_b^*(k)$ by (\ref{eq:5.1}) and send $\hat u_s(k)$, $u_b^*(k)$ to PV generators and batteries.

\\\hrulefill
\State  \textbf{3. Estimation}

\\1) Calculate $L(k)$ by solving the minimization problem (\ref{op:3.1})

\\2) Estimate the current state of the batteries $\hat x(k)$ by (\ref{eq:3.8}) with the measurement $y(k)$, $u_b^*(k-1)$, $\hat x(k-1)$ and $L(k)$.

\\\hrulefill
\State  \textbf{4. Prediction}

\\1) Estimate the state of the next step $\hat x(k+1)$ with the obtained $u_b^*(k)$ and the estimated state $\hat x(k)$;

\\2) Obtain the updated forecasts of the future PV generation power $\hat P_s^a(k+1), \cdots, \hat P_s^a(k+N)$ and the load demand $\hat P_l(k+1), \cdots, \hat P_l(k+N)$, and calculate the control input prediction $\hat u_s(k+1) \cdots \hat u_s(k+N)$ and $\hat u_b(k+1) \cdots \hat u_b(k+N)$ by solving the optimization problem~(\ref{op:4.1}).

\\3) Set $u_s^d(k+n)\leftarrow \hat{u}_s(k+n)$, and $u_b^d(k+n)\leftarrow \hat{u}_b(k+n), n=1, \cdots, N$.

\\\hrulefill
\State  \textbf{5. Loop}\\

\quad Set $k\leftarrow k+1$, and go to \textit{Step\ 2}.

\end{algorithmic}
\end{algorithm}

Consequently, as illustrated in Fig. \ref{fig:4.1}, to reduce the computational delay at every time step, a one-step-ahead method is adopted, and the control mechanism with the deviation compensation can be generalized as follows.
At every sampling time step, the central controller collects information from local units. With this information, a compensation term is computed and added to the control inputs predicted at the last time step. Then, the obtained optimal control inputs are implemented in the system.
During the interval between the two sampling steps, the central controller first estimates the current system state by solving the SME optimization problem.
Next, it predicts the future control input sequence by solving the MPC optimization problem and storing the control sequence for future use.
With this procedure, most of the time-consuming computations are executed between the sampling steps; only the deviation compensation must be computed immediately at the sampling step. As a result, the computational delay is effectively relieved.
Moreover, because of the large timescale of the tertiary power flow control, its sampling time is relatively long (usually minutes to hours), which allows plenty of time for computing estimation and prediction.

The corresponding control algorithm is designed as in Algorithm~\ref{alg:1}.

\section{Simulation Results}
In simulation tests, the Raspberry Pi $4$ Model B is used as the local hub to receive and send data wirelessly, and a desktop PC is used to execute all the computations, with MATLAB R2016b implemented as the central controller, as illustrated in Fig. \ref{fig:5.10}. The SME optimization problem is solved with MATLAB YALMIP \cite{lofberg2004cacsd} toolbox and SeDuMi \cite{jos1999oms} solver, while the MPC optimization problem is solved with Gurobi \cite{gurobi2016} solver.

\begin{figure}[htbp]
  \centering
  \includegraphics[width=0.4\textwidth]{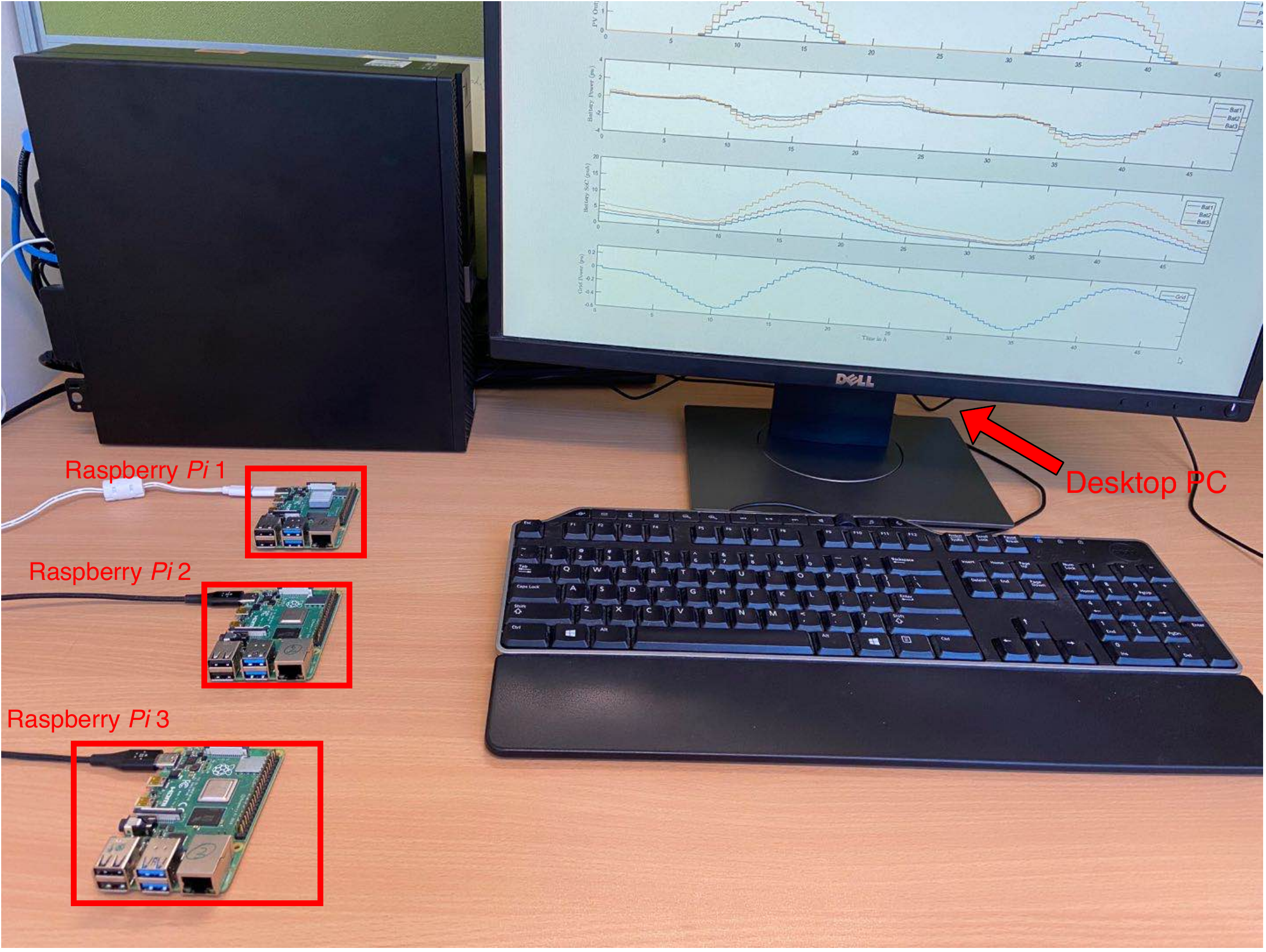}
  \caption{Raspberry $Pi$s}\label{fig:5.10}
\end{figure}

\begin{figure}[htbp]
  \centering
  \includegraphics[width=0.38\textwidth]{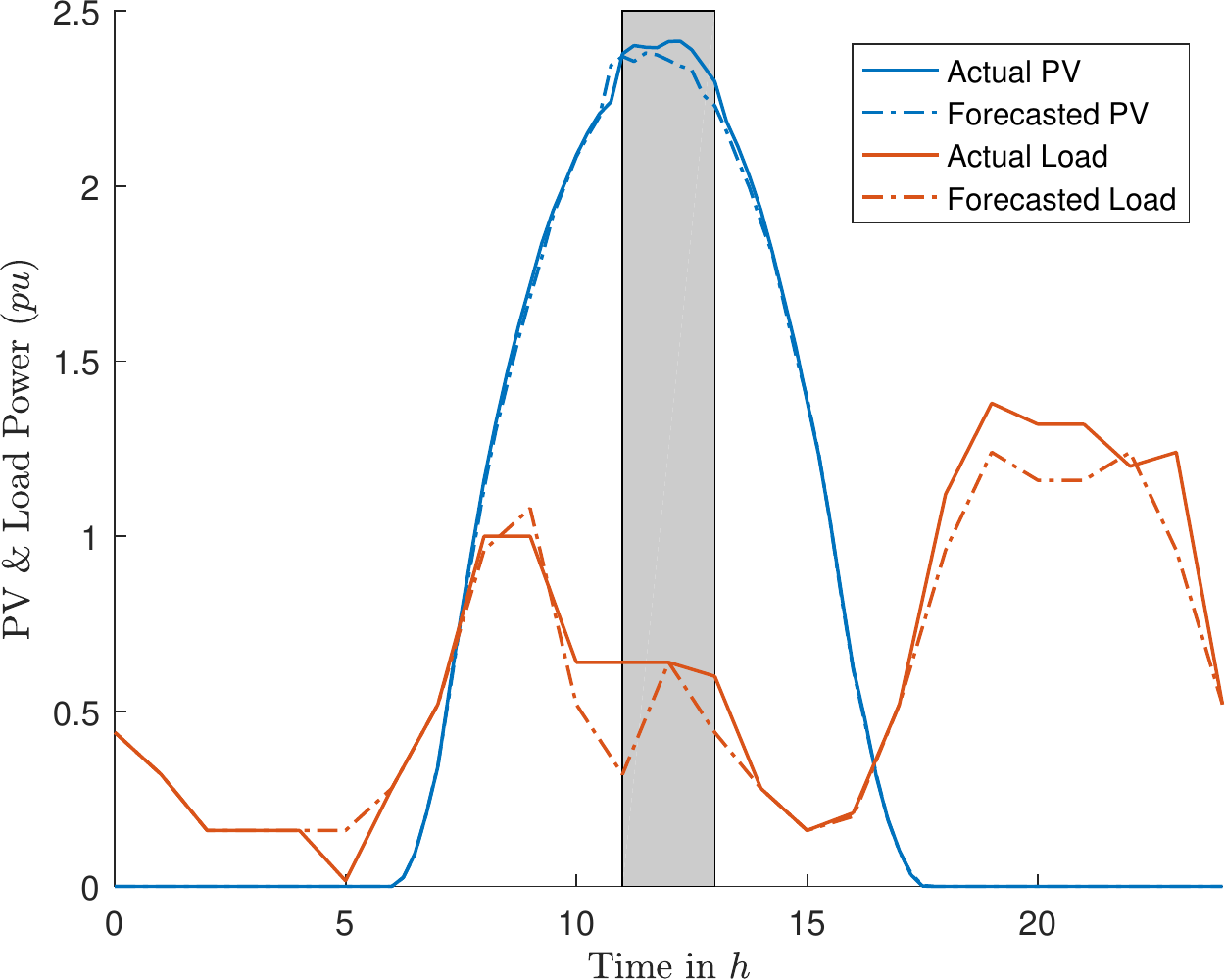}
  \caption{PV generation \& load demand profiles}\label{fig:5.1}
\end{figure}

The PV generation profile was collected from the UQ Centre at The University of Queensland on September 1, 2019~\cite{uqpv}. The load profile is the power consumption data for a typical residential household.
As shown in Fig.~\ref{fig:5.1}, the forecasts for both PV generation and load demand are displayed with dashed lines, while the evolutions of the actual PV generation power and load power are displayed with solid lines.

To simplify the data composition, it is assumed that the three PV generators have the same evolution pattern, as shown in Fig. \ref{fig:5.1}, with the only difference being in the value magnitudes, which are proportional to one another at every time step; the three residential loads also evolve proportionally in the same pattern. Specifically, in the test setting, the power values of all the PV and load units have the following proportional relationships: $P_{pv1}=0.5P_{pv}$, $P_{pv2}=P_{pv}$, $P_{pv3}=1.5P_{pv}$, and $P_{load1}=0.5P_{load}$, $P_{load2}=P_{load}$, $P_{load3}=1.8P_{load}$, where $P_{pv}$ and $P_{load}$ are the default reference values corresponding to Fig.~\ref{fig:5.1}.
The electrical fault or communication failure is supposed to occur during the time $11h$ to $13h$, as marked in the shadow box in the figure.
The detailed parameter values are given in Table \ref{tab:3.2}.
\begin{table}[!t]
	\renewcommand{\arraystretch}{1.3}
	\caption{Unit parameters and weights of the microgrids}\label{tab:3.2}
	\centering
	\resizebox{0.95\columnwidth}{!}{
	\begin{tabular}{l l l}
\hline\hline \\[-3mm]
  \multicolumn{1}{c}{Parameter} & \multicolumn{1}{c}{Description} & \multicolumn{1}{c}{Value}
  \\[1.6ex] \hline
$[P_s^{min} \quad P_s^{max}]$ & PV generation limits  & $\left[\begin{matrix}
                                                               0 & 1.5 \\
                                                               0 & 3 \\
                                                               0 & 4.5
                                                             \end{matrix}\right]pu$ \\
\hline
$[P_b^{min} \quad P_b^{max}]$ & Battery power limits  & $\left[\begin{matrix}
                                                               -3 & 3 \\
                                                               -4 & 4 \\
                                                               -6 & 6
                                                             \end{matrix}\right]pu$ \\
\hline
$[x^{min} \quad x^{max}]$ & Battery SoC limits  & $\left[\begin{matrix}
                                                               0 & 12 \\
                                                               0 & 16 \\
                                                               0 & 24
                                                             \end{matrix}\right]puh$\\
\hline
$[x_b^{min} \quad x_b^{max}]$ & Battery SoC thresholds  & $\left[\begin{matrix}
                                                               0.2 & 11.8 \\
                                                               0.3 & 15.7 \\
                                                               0.3 & 23.7
                                                             \end{matrix}\right]puh$ \\
\hline
$x(0)$ & Battery initial state & $\left[\begin{matrix}
                                          3 & 4 & 6
                                        \end{matrix}\right]^T puh$ \\
\hline
$C_s$ & Weights for solar units & $\left[\begin{matrix}
                                          1 & 1 & 1
                                        \end{matrix}\right]^T 1/{pu}^2$ \\
\hline
$C_{b,1}$ & Weights for battery power & $\left[\begin{matrix}
                                          0.2 & 0.15 & 0.1
                                        \end{matrix}\right]^T 1/{pu}^2$\\
\hline
$C_{b,2}$ & Weights for battery SoC & $\left[\begin{matrix}
                                          0.3 & 0.3 & 0.3
                                        \end{matrix}\right]^T 1/{puh}^2$\\
\hline
$C_{g,1}$ & Weight 1 for power grid  & $0.5 \quad 1/{pu}^2$ \\
\hline
$C_{g,2}$ & Weight 2 for power grid & $0.1 \quad 1/{pu}$ \\
\hline
$\lambda_b$ & Coefficient vector & $\left[\begin{matrix}
                                          1 & 1 & 1
                                        \end{matrix}\right]^T$ \\
\hline\hline
\end{tabular}
}
\end{table}

\begin{figure}[htbp]
  \centering
  \includegraphics[width=0.48\textwidth]{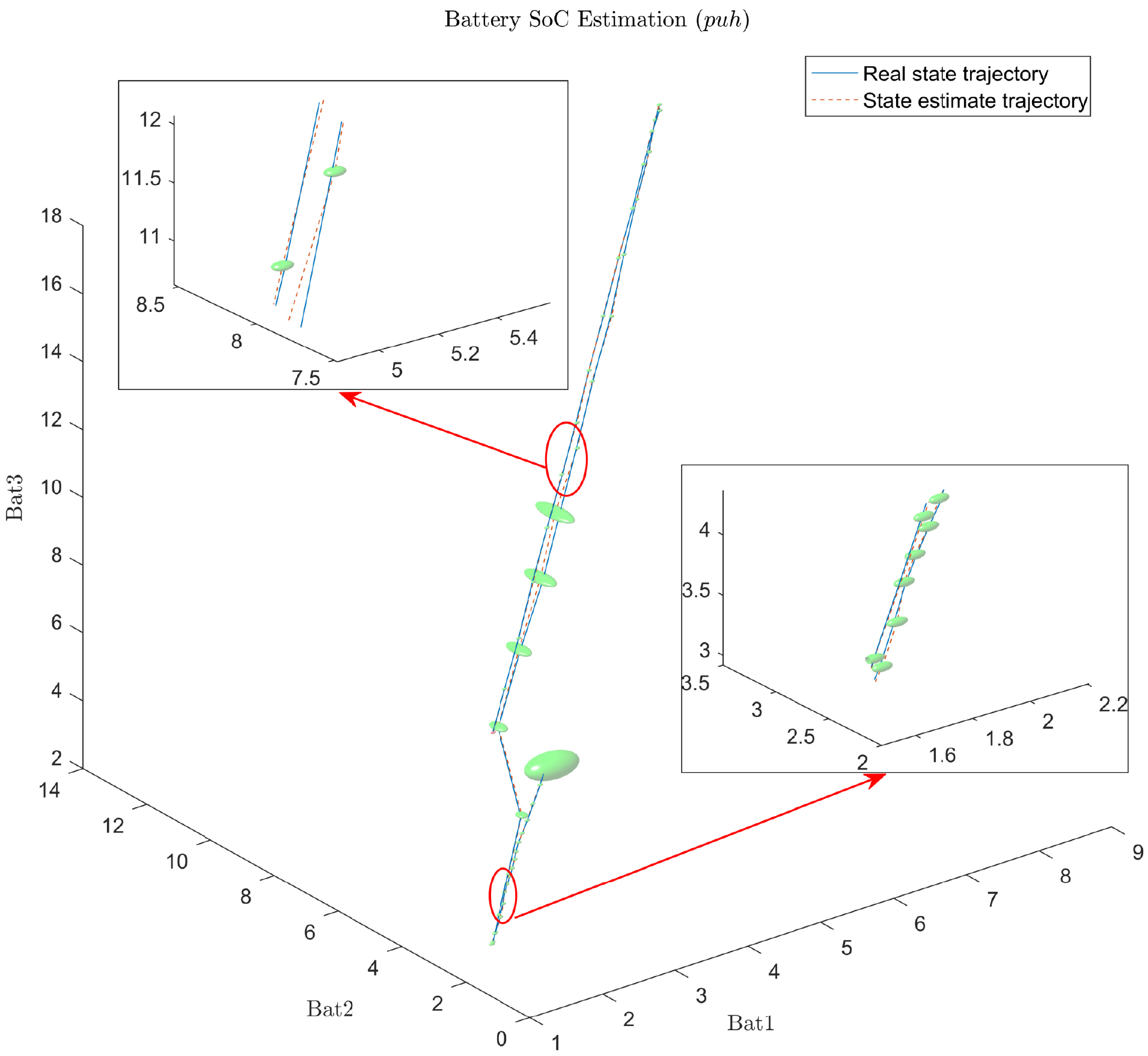}
  \caption{State trajectories and estimation ellipsoids}\label{fig:5.7}
\end{figure}

%
%

\begin{figure}[htbp]
  \centering
  \includegraphics[width=0.48\textwidth]{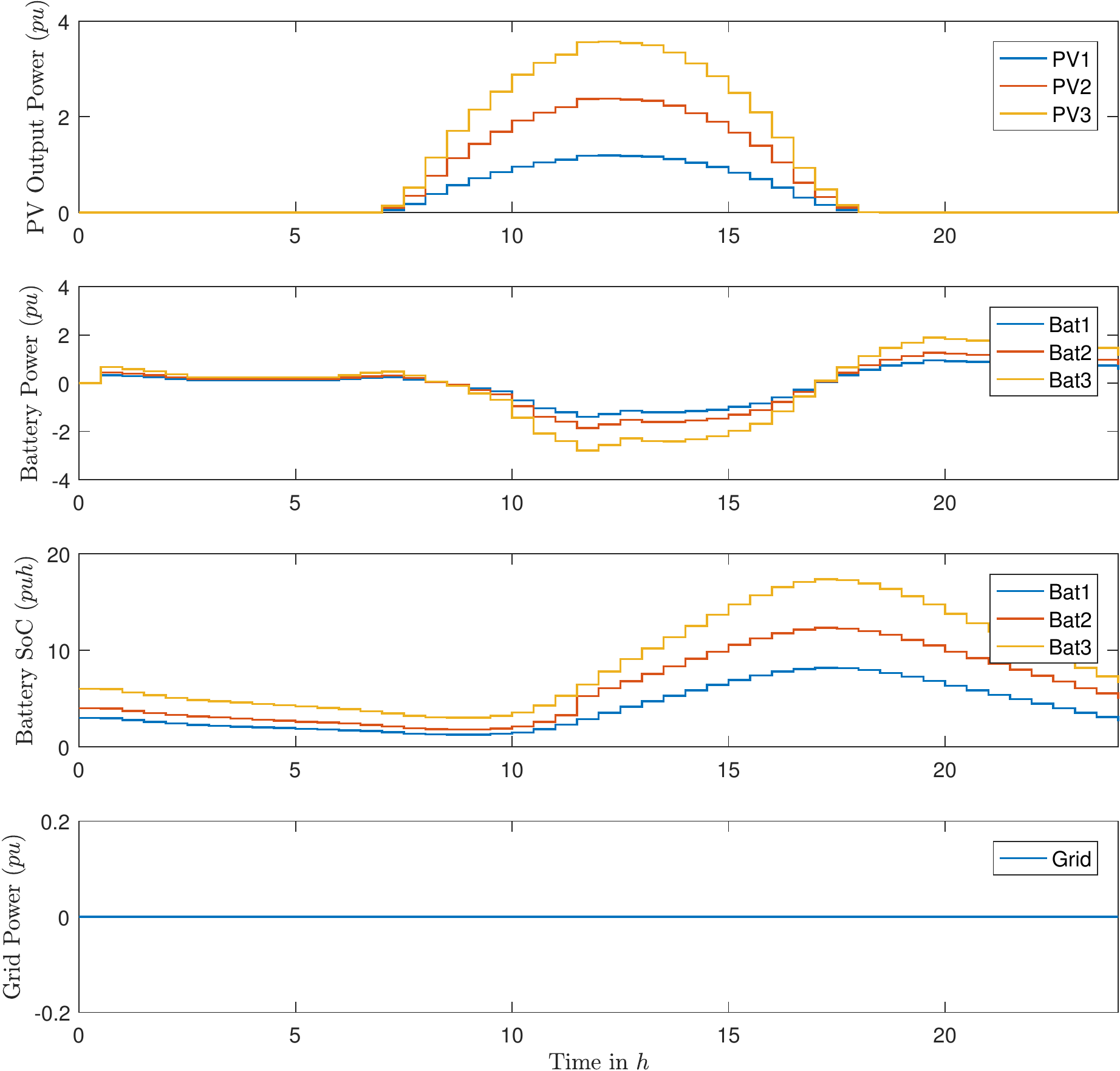}
  \caption{System responses in islanded mode}\label{fig:5.9}
\end{figure}

\vspace{1mm}
\textbf{Case 1 } \textit{Ellipsoidal state estimation test in islanded mode}

As in the setting, communication failure occurs in Battery 2 from time $11h$ to $13h$, which means the communication vector $\mathcal{A}_b(k)$ switches from
$\left[\begin{matrix}
 1 & 1 & 1
 \end{matrix}\right]^T$ to
 $\left[\begin{matrix}
  1 & 0 & 1
  \end{matrix}\right]^T$ during this period. In this scenario, the state of Battery 2 is unavailable during the communication failure and must be estimated using the designed SME method.

It is assumed that the forecasts for both the PV generation power and load power are accurate in this case and that the process disturbance and measurement noise satisfy the following bounding conditions: $\omega(k)\in W(k)=\{\omega(k)\in\mathbb{R}:||\omega(k)||_{\infty}\leq0.1\ pu\}$ and $\upsilon(k)\in V(k)=\{\upsilon(k)\in\mathbb{R}:||\upsilon(k)||_{\infty}\leq0.02\ puh\}$. The initial conditions are given as $x(0)=\left[\begin{matrix}
                                            3 & 4 & 6
                                          \end{matrix}\right]^T puh$,
$\hat x(0)=\left[\begin{matrix}
3.1 & 4.1 & 5.8
\end{matrix}\right]^T puh$, and  $u(0)=\left[\begin{matrix}
0 & 0 & 0
\end{matrix}\right]^T$. Hence, with conservative ellipsoidal approximations, the unknown disturbance and noise are bounded with $Q=\left[\begin{matrix}
                                                 0.03 & 0 & 0 \\
                                                 0 & 0.03 & 0 \\
                                                 0 & 0 & 0.03
                                               \end{matrix}\right]$ and $R=\left[\begin{matrix}
                                                 0.0012 & 0 & 0 \\
                                                 0 & 0.0012 & 0 \\
                                                 0 & 0 & 0.0012
                                               \end{matrix}\right]$, and the initial system state estimation error is bounded with $P(0)=\left[\begin{matrix}
                                                 0.12 & 0 & 0 \\
                                                 0 & 0.12 & 0 \\
                                                 0 & 0 & 0.12
                                               \end{matrix}\right]$.

The test results are shown in Fig.~\ref{fig:5.7} and~\ref{fig:5.9}. From Fig.~\ref{fig:5.7}, it can be observed that the state estimates are the centers of the ellipsoids and the real states reside in the ellipsoids all the time. In addition, the estimation ellipsoids converge in small volumes except for those time steps when communication failures occur. This is because the state of Battery 2 is unavailable during this period, and the corresponding element in the measurement output is $0$, as in (\ref{eq:3.1}). Hence, the state estimation of the current step is obtained without the exact measurement update from Battery 2, as in (\ref{eq:3.7}). All the batteries work properly to store and supply power such that the global power balance is satisfied, as in Fig.~\ref{fig:5.9}.

\begin{figure}[htbp]
  \centering
  \includegraphics[width=0.48\textwidth]{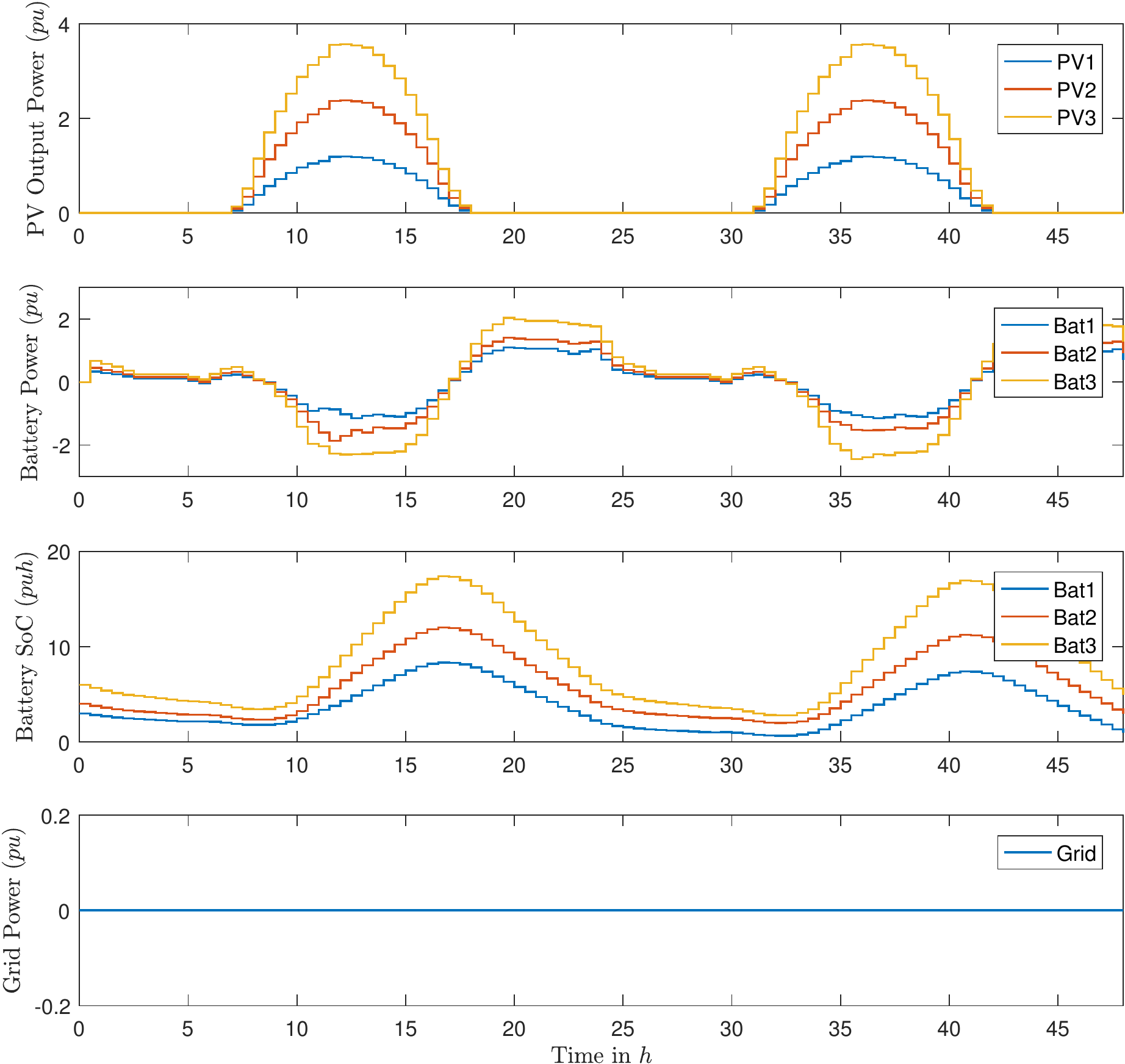}
  \caption{System responses with deviation compensation}\label{fig:5.4}
\end{figure}

\begin{figure}[htbp]
  \centering
  \includegraphics[width=0.48\textwidth]{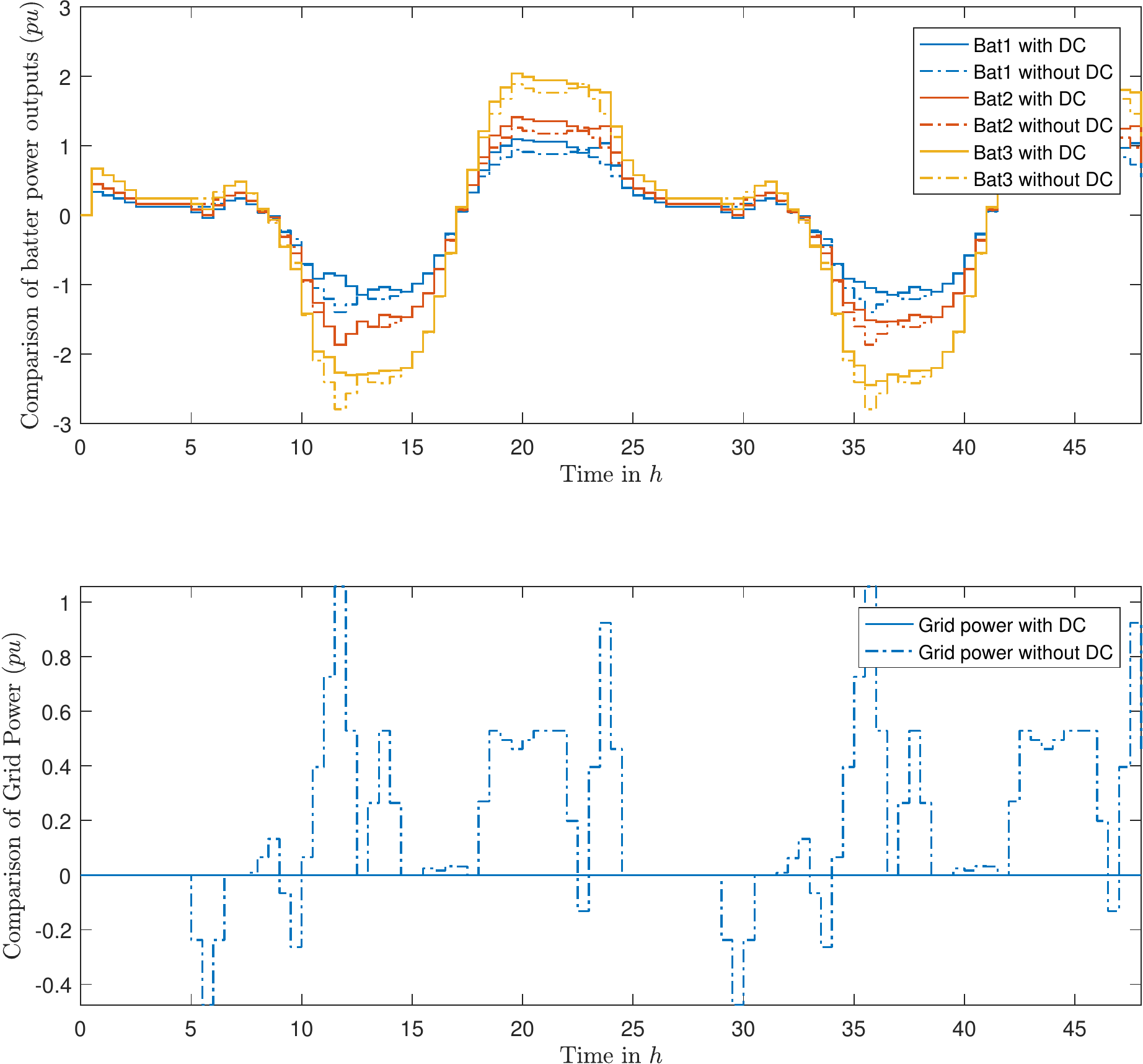}
  \caption{Power values of batteries and utility grid with and without deviation compensation}\label{fig:5.5}
\end{figure}

\vspace{1mm}
\textbf{Case 2 } \textit{Deviation compensation test}

In this case, the designed deviation compensation algorithm is implemented and compared with the scenario where there is no deviation compensation. For simplicity, the utility grid power is expected to be $0$, which is equivalent to the islanded mode. As shown in the system responses in Fig.~\ref{fig:5.4}, the designed controller can still regulate the batteries properly even when there are forecast errors. In addition, from Fig.~\ref{fig:5.5}, it can be seen that the grid power will experience severe fluctuations if there is no compensation procedure, because the power deviation in the system must be compensated by the grid.
By contrast, if the proposed compensation method is used, this type of deviation can be compensated by the batteries such that the utility grid power will be smoothed as the predefined constant.

\begin{figure}[htbp]
  \centering
  \includegraphics[width=0.48\textwidth]{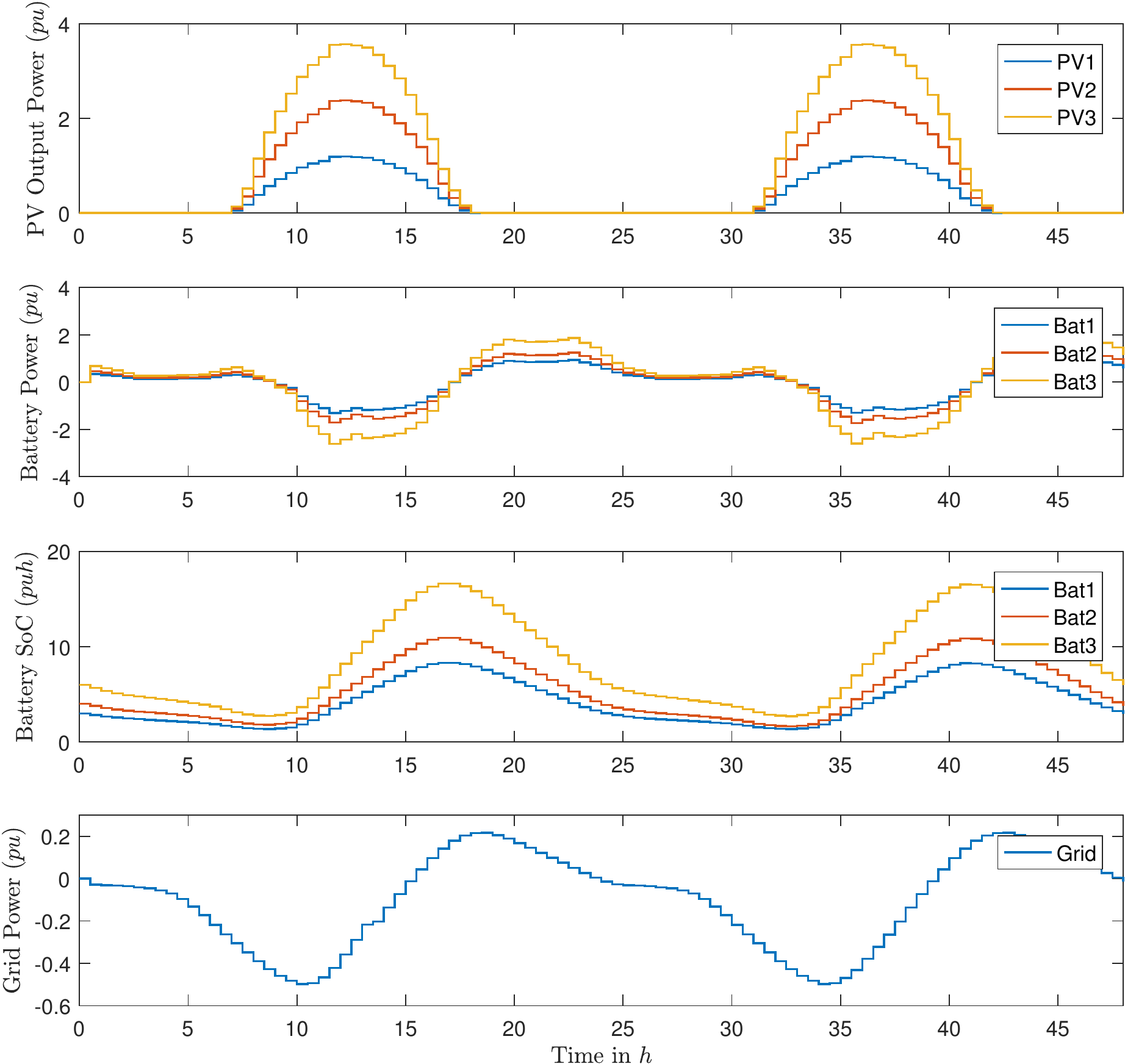}
  \caption{System responses with variable grid power}\label{fig:5.2}
\end{figure}

\begin{figure}[htbp]
  \centering
  \includegraphics[width=0.48\textwidth]{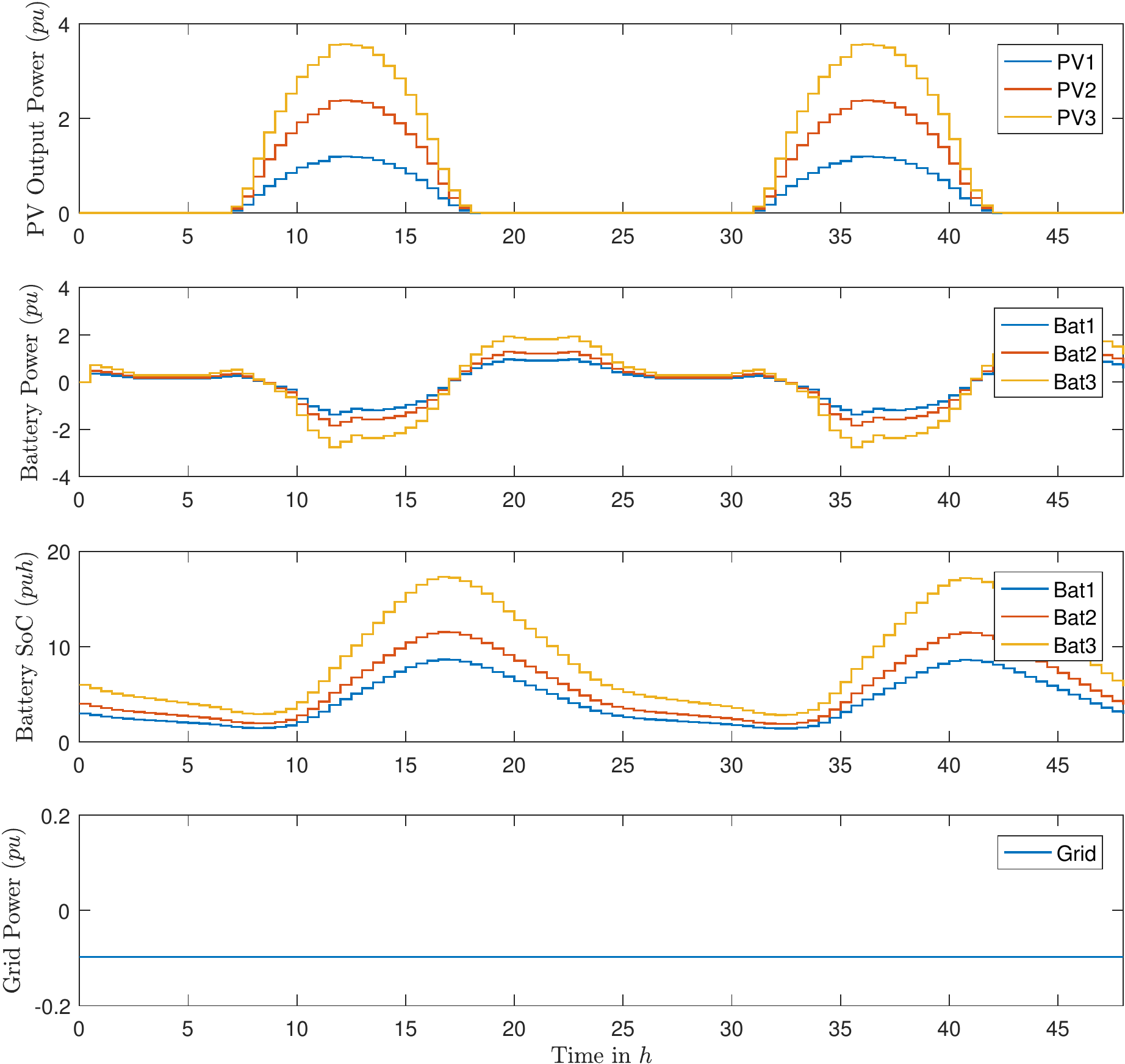}
  \caption{System responses with constant grid power}\label{fig:5.3}
\end{figure}

\vspace{1mm}
\textbf{Case 3 } \textit{MPD in grid-connected mode with communication failure}

In this case, the MMS is set to work in grid-connected mode, and the designed MPD strategy is implemented in the system while still considering the occurrence of communication failure in Battery 2.

First, the scenario wherein the power of the utility grid is allowed to be variable is tested, and the simulation results are illustrated in Fig.~\ref{fig:5.2}.
The results show that the controller can still regulate the PV and battery power properly. However, the power fed into or supplied by the grid fluctuates as the generation power of the solar panels changes. This is because the global system power balance must be satisfied at all times.

Then, the power fed into the utility grid is predefined as a constant of $-0.0979 pu$. This constant is obtained by computing the average value of the variable power profile of the utility grid from the last test.
As shown in Fig.~\ref{fig:5.3}, with a properly predefined power value for the utility grid, the charging/discharging power of the batteries can still be regulated effectively.

\begin{figure}[htbp]
  \centering
  \includegraphics[width=0.48\textwidth]{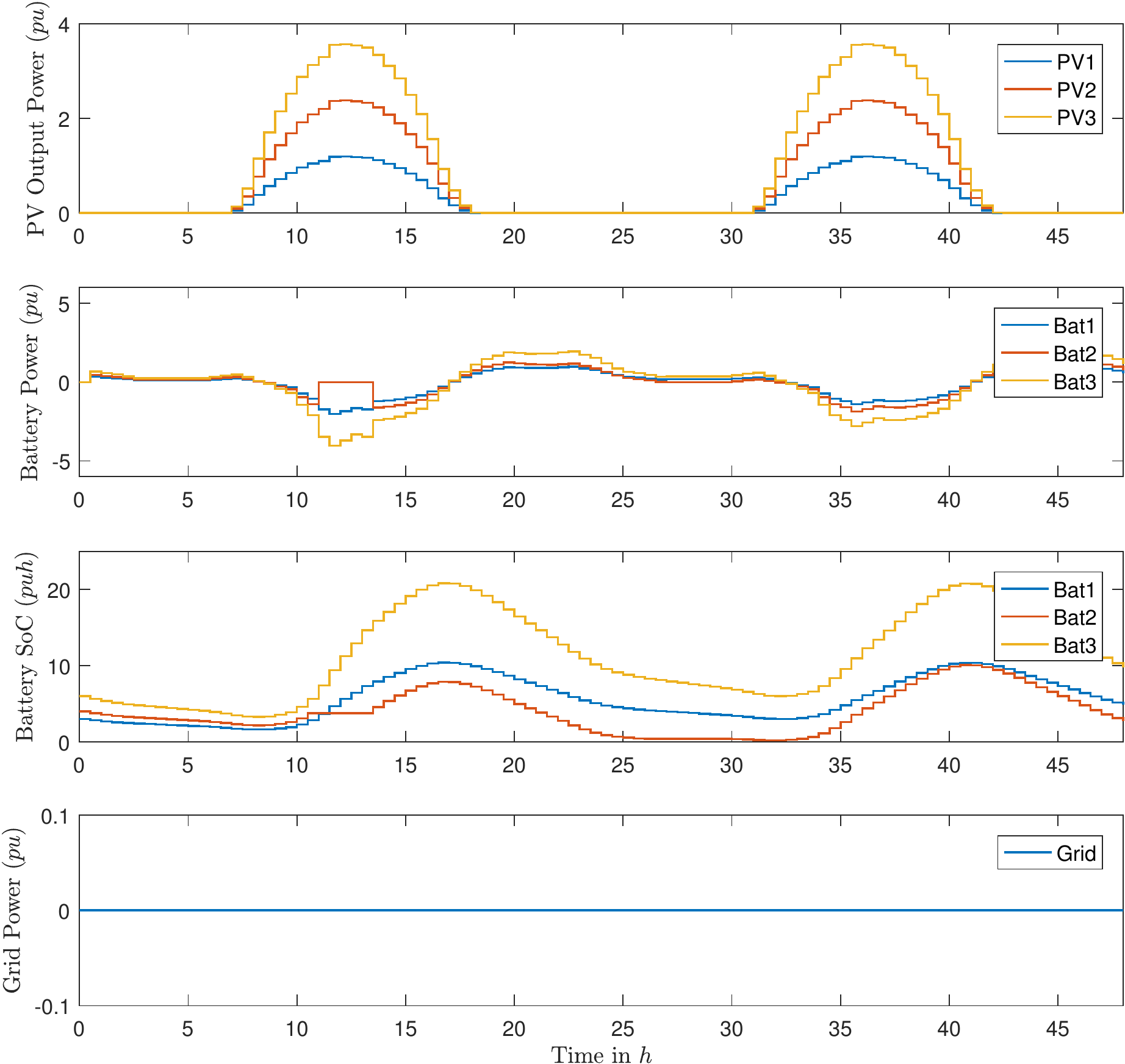}
  \caption{System responses with electrical disconnection}\label{fig:5.6}
\end{figure}

\vspace{1mm}
\textbf{Case 4 } \textit{MPD in islanded mode with electrical fault}

For further demonstration, an electrical fault is considered in the MMS between the times $11h$ and $13h$ at Battery 2. To verify that the designed MPD strategy can effectively manage the power flow within the local system, the MMS is set to work in islanded mode to exclude the influence from the external utility grid.
As shown in Fig.~\ref{fig:5.6}, the power of Battery 2 becomes $0$ and its SoC keeps invariant during the period in which the electrical fault occurs. However, the whole system can still work properly, as the other two batteries can regulate their own power values to satisfy the power balance requirement, even though Battery 2 is disconnected. This demonstrates that the proposed MPD strategy works effectively for the situation with electrical faults.

\section{Conclusion}

Based on the simulation test results, it can be concluded that the proposed control strategy can address the power dispatch problem for faulty microgrids in both islanded mode and grid-connected mode, and that the designed deviation compensation method can effectively compensate for forecast errors by rescheduling the battery charge/discharge.
Future research should focus on implementing the designed wireless communication structure in a prototype microgrid system and using Raspberry $Pi$s to provide set-point control signals and collect information for the central energy management operator.

\appendix

\emph{Proof of Theorem~\ref{th:3.1}:}

  With the system state $x(k)$ in (\ref{eq:3.5}) and the state estimate $\hat{x}(k)$ in (\ref{eq:3.8}), the state error can be derived as
  \vspace{-0.2cm}
  \begin{align}\label{eq:3.10}
    &x(k)-\hat{x}(k)\nonumber\\
    =&x(k-1)+B(k-1)u(k-1)+F(k-1)\omega(k-1)\nonumber\\
    &+\delta(k-1)-[(I-L(k)C(k))\hat{x}({k-1})\nonumber\\
    &+(I-L(k)C(k))B({k-1})u({k-1})+L(k)y(k)\nonumber\\
    &+(I-L(k)C(k))\delta(k-1)]
  \end{align}

  By substituting $y(k)=C(k)x(k)+D(k)\upsilon(k)$ into (\ref{eq:3.10}), the state error can be rewritten as
  \vspace{-0.2cm}
  \begin{align}\label{eq:3.11}
    &x(k)-\hat{x}(k)\nonumber\\
    =&(I-L(k)C(k))(x(k-1)-\hat{x}(k-1))+(I-L(k)C(k))\nonumber\\
     &F(k-1)\omega(k-1)-L(k)D(k)\upsilon(k).
  \end{align}

  Considering the assumption that $x(k-1)\in \varepsilon (\hat{x}(k-1), P(k-1))$, which implies that\vspace{-0.2cm}
  \begin{multline}\label{eq:3.12}
    (x({k-1})-\hat{x}({k-1}))^TP^{-1}({k-1})(x({k-1})-\hat{x}({k-1}))\\
    \leq 1,
  \end{multline}
then there exists a $z({k-1})$ with $||z({k-1})||\leq1$ such that
\begin{align}\label{eq:3.13}
  x({k-1})=\hat{x}({k-1})+E({k-1})z({k-1}),
\end{align}
where $E({k-1})$ is the Cholesky factorization of $P({k-1})$, i.e.,
$P({k-1})=E({k-1})E^T({k-1})$.

Hence, by substituting (\ref{eq:3.13}) into (\ref{eq:3.11}), then the following equation is obtained\vspace{-0.2cm}
  \begin{align}\label{eq:3.14}
    &x(k)-\hat{x}(k)\nonumber\\
    =&(I-L(k)C(k))E({k-1})z({k-1})+(I-L(k)C(k))\nonumber\\
     &F(k-1)\omega(k-1)-L(k)D(k)\upsilon(k).
  \end{align}

Denote $\eta({k})=[z({k-1})^T\ \omega({k-1})^T\ \upsilon(k)^T\ 1]^T$, then (\ref{eq:3.14}) can be rewritten in a compact form\vspace{-0.2cm}
\begin{align}\label{eq:3.15}
  x(k)-\hat{x}(k)=\Phi_{\eta}(k)\eta({k}).
\end{align}

With the equation (\ref{eq:3.15}), the current state estimation ellipsoid $\varepsilon(\hat{x}(k),P(k))$ can be written as
\begin{align}\label{eq:3.16}
  \eta({k})^T[\Phi_{\eta}(k)^TP(k)^{-1}\Phi_{\eta}(k)-\textnormal{diag}(0,0,0,1)]\eta({k})\leq0.
\end{align}

Since the unknown variables $z({k-1})$, $\omega({k-1})$, $\upsilon(k)$ within $\eta(k)$ satisfy the following conditions\vspace{-0.2cm}
\begin{align*}
  \left\{
  \begin{aligned}
  &||z({k-1})||\leq1\\
  &\omega({k-1})^TQ({k-1})^{-1}\omega({k-1})\leq1\\
  &\upsilon(k)^TR(k)^{-1}\upsilon(k)\leq1
  \end{aligned}
  \right.,
\end{align*}
which can be also expressed in the form of $\eta(k)$ as
\begin{align}\label{eq:3.17}
  \left\{
  \begin{aligned}
  &\eta^T(k)\textnormal{diag}(I,0,0,-1)\eta(k)\leq0\\
  &\eta^T(k)\textnormal{diag}(0,Q({k-1})^{-1},0,-1)\eta(k)\leq0\\
  &\eta^T(k)\textnormal{diag}(0,0,R(k)^{-1},-1)\eta(k)\leq0
  \end{aligned}
  \right.,
\end{align}
by applying S-procedure to (\ref{eq:3.16}) and (\ref{eq:3.17}), then the inequality (\ref{eq:3.16}) can hold if there exist nonnegative scalars $\lambda_1$, $\lambda_2$ and $\lambda_3$ such that\vspace{-0.2cm}
\begin{align}\label{eq:3.18}
 &\eta^T(k)(\Phi_{\eta}(k)^TP(k)^{-1}\Phi_{
  \eta}(k)-\textnormal{diag}(0,0,0,1)
  -\lambda_1\textnormal{diag}(I,\nonumber\\&0,0,-1)
  -\lambda_2\textnormal{diag}(0,Q({k-1})^{-1},0,-1)
  -\lambda_3\textnormal{diag}(0,0,\nonumber\\&R(k)^{-1},-1))\eta(k)\leq0.
\end{align}

By denoting\\
$\Psi_{\lambda_1,\lambda_2,\lambda_3}=\textnormal{diag}(\lambda_1I, \lambda_2Q^{-1}({k-1}),\lambda_3R^{-1}(k), 1-\lambda_1-\lambda_2-\lambda_3)$,\\
the inequality (\ref{eq:3.18}) can be rewritten as
\begin{align}\label{eq:3.19}
  \eta^T(k)(\Phi_{\eta}(k)^TP(k)^{-1}\Phi_{\eta}(k)-\Psi_{\lambda_1,\lambda_2,\lambda_3})\eta(k)\leq0.
\end{align}

Through applying Schur complements, the inequality (\ref{eq:3.19}) is equivalent to (\ref{eq:3.9}).

Thereby, the proof is completed.

\bibliographystyle{IEEEtran}
\bibliography{D:/Document/Latex/bib/abrv,D:/Document/Latex/bib/conf,D:/Document/Latex/bib/thesis,D:/Document/Latex/bib/article}

\end{document}